\documentclass[pra,twocolumn,a4paper,superscriptaddress]{revtex4}
\usepackage{graphicx,amsmath,bm,natbib, dsfont, longtable, amssymb}
\usepackage[bbgreekl]{mathbbol}
\usepackage{bm}
\usepackage{amsmath}
\usepackage{amsthm}
\usepackage{graphicx}
\usepackage{amsfonts}
\usepackage{amssymb}
\usepackage{float}
\usepackage{xcolor}
\usepackage{url}

\newtheorem{TH1}{Theorem}
\newtheorem{COR}{Corollary}[TH1]

\def\ket#1{\mathinner{|{#1}\rangle}}

\newcommand{\ketbra}[2]{|#1\rangle\!\!\!\;\langle#2|}

\def\text#1{\textrm{#1}}

\newcommand{\C}{\mathcal{C}}

\begin{document}

\title{Bell correlations in a many-body system with finite statistics}
\author{Sebastian Wagner}
\affiliation{Quantum Optics Theory Group, Department of Physics, University of Basel, Klingelbergstrasse 82, 4056 Basel, Switzerland}
\author{Roman Schmied}
\affiliation{Quantum Atom Optics Lab, Department of Physics, University of Basel, Klingelbergstrasse 82, 4056 Basel, Switzerland}
\author{Matteo Fadel}
\affiliation{Quantum Atom Optics Lab, Department of Physics, University of Basel, Klingelbergstrasse 82, 4056 Basel, Switzerland}
\author{Philipp Treutlein}
\affiliation{Quantum Atom Optics Lab, Department of Physics, University of Basel, Klingelbergstrasse 82, 4056 Basel, Switzerland}
\author{Nicolas Sangouard}
\affiliation{Quantum Optics Theory Group, Department of Physics, University of Basel, Klingelbergstrasse 82, 4056 Basel, Switzerland}
\author{Jean-Daniel Bancal}
\affiliation{Quantum Optics Theory Group, Department of Physics, University of Basel, Klingelbergstrasse 82, 4056 Basel, Switzerland}

\date{\today}

\begin{abstract}
\noindent A recent experiment reported the first violation of a Bell correlation witness in a many-body system [Science 352, 441 (2016)]. Following discussions in this paper, we address here the question of the statistics required to witness Bell correlated states, i.e. states violating a Bell inequality, in such experiments.
We start by deriving multipartite Bell inequalities involving an arbitrary number of measurement settings, two outcomes per party and one- and two-body correlators only.
Based on these inequalities, we then build up improved witnesses able to detect Bell-correlated states in many-body systems using two collective measurements only. These witnesses can potentially detect Bell correlations in states with an arbitrarily low amount of spin squeezing. We then establish an upper bound on the statistics needed to convincingly conclude that a measured state is Bell-correlated.
\end{abstract}
\maketitle

\paragraph{Introduction --}

Physics research fundamentally relies on the proper analysis of finite experimental data. In this exercise, assumptions play a subtle but crucial role. On the one hand, they are needed in order to reach a conclusion; even device-independent assessments rely on assumptions~\cite{Scarani12}.
On the other hand, they open the door for undesirable effects ranging from a reduction of the conclusion's scope, when more assumptions are used than strictly needed, to biased results when relying on unmet assumptions. Contrary to popular belief, such cases are frequent in science, even for common assumptions~\cite{Bailey16,Youden72}. Relying on fewer hypotheses, when possible, is thus desirable to obtain more general, accurate and trustworthy conclusions~\cite{diew1,diew2}.

Bell nonlocality, as revealed by the violation of a Bell inequality, constitutes one of the strongest forms of non-classicality known today. However, its demonstration has long been restricted to systems involving few particles~\cite{Lanyon14,Eibl02,Zhao03,Lavoie09,BellExp}. Recently, the discovery of multipartite Bell inequalities that only rely on one- and two-body correlators opened up new possibilities~\cite{Tura14}. Although these inequalities have not yet lead to the realization of a multipartite Bell test, they can be used to derive witnesses able to detect Bell correlated states, i.e. states capable of violating a Bell inequality.

Using such a witness, an experiment recently detected the presence of Bell correlations in a many-body system under the assumption of gaussian statistics~\cite{Schmied16}. While this demonstration uses spin squeezed states, the detection of Bell correlations in other systems was also recently investigated~\cite{Pelisson16}.
The witness used in Ref.~\cite{Schmied16} involves one- and two-body correlation functions and takes the form $\mathcal{W} \geq 0$, where the inequality is satisfied by measurements on states that are not Bell-correlated. Observation of a negative value for $\mathcal{W}$ then leads to the conclusion that the measured system is Bell-correlated. However, due to the statistics loophole~\cite{Gill12,Zhang11}, reaching such a conclusion in the presence of finite statistics requires special care.
In particular, an assessment of the probability with which a non-Bell-correlated state could be responsible for the observed data is required before concluding about the presence of Bell correlations without further assumptions.

Concretely, the witness of Refs.~\cite{Schmied16} has the property of admitting a quantum violation lower-bounded by a constant $\mathcal{W}_\text{opt} < 0$, while the largest possible value $\mathcal{W}_\text{max} > 0$ is achievable by a product state and increases linearly with the size of the system $N$. These properties imply that a small number of measurements on a state of the form
\begin{equation}\label{eq:counterExample}
\rho=(1-q)\ketbra{\psi}{\psi} + q(\ketbra{{\uparrow}}{{\uparrow}})^{\otimes N},
\end{equation}
where $\mathcal{W}(\ket{\psi})=\mathcal{W}_\text{opt}$, $\mathcal{W}(\ket{\uparrow}^{\otimes N})=\mathcal{W}_\text{max}$ and $q$ is small, is likely to produce a negative estimate of $\mathcal{W}$, even though the state is not detected by the witness in the limit of infinitely many measurement rounds~\cite{Schmied16}. This state thus imposes a lower bound on the number of measurements required to exclude, through such witnesses, all non-Bell-correlated states with high confidence. Contrary to other assessments, this lower bound increases with the number of particles involved in the many-body system. Therefore, it is not captured by the standard deviation of one- and two-body correlation functions (which on the contrary decreases as the number of particles increases).

It is worth noting that states of the form~\eqref{eq:counterExample} put similar bounds on the number of measurements required to perform any hypothesis tests in a many-body system satisfying the conditions above. This includes in particular tests of entanglement~\cite{Gross10,Riedel10,Bernd14,Pezze} based on the entanglement witnesses of Ref.~\cite{Sorensen01a,Toth09,Hyllus12}.

In this article, we address this statistical problem in the case of Bell correlation detection by providing a number of measurement rounds sufficient to exclude non-Bell-correlated states from an observed witness violation. Let us mention that in Ref.~\cite{Schmied16}, the statistics loophole is circumvented by the addition of an assumption on the set of local states being tested. This has the effect of reducing the scope of the conclusion: the data reported in Ref.~\cite{Schmied16}, are only able to exclude a subset of all non-Bell-correlated states (as pointed out in the reference).
Here, we show that such additional assumptions are not required in experiments on many-body systems, and thus argue that they should be avoided in the future.

In order to minimize the amount of statistics required to reach our conclusion, we start by investingating improved Bell correlation witnesses. For this, we first derive Bell inequalities with two-body correlators and an arbitrary number of settings. This allows us to obtain Bell-correlation witnesses that are more resistant to noise compared to the one known to date~\cite{Schmied16}. We then analyse the statistical properties of these witnesses and provide an upper bound on the number of measurement rounds needed to rule out all local states in a many-body system. We show that this upper bound is linear in the number of particles, hence making the detection of Bell correlations free of the statistical loophole possible in systems with a large number of particles.

\paragraph{Symmetric two-body correlator Bell inequalities with an arbitrary number of settings --}
Multipartite Bell inequalities that are symmetric under exchange of parties and which involve only one- and two-body correlators have been proposed in scenarios where each party uses two measurement settings and receives an outcome among two possible results \cite{Tura14}. Similar inequalities were also obtained for translationally invariant systems~\cite{Tura14a}, or based on Hamiltonians~\cite{Tura16}. Here, we derive a similar family of Bell inequalities that is invariant under arbitrary permutations of parties but allows for an arbitrary number of measurement settings per party.

Let us consider a scenario in which $N$ parties can each perform one of $m$ possible measurements $M_k^{(i)}$ ($k=0,...,m-1$; $i=1,...,N$) with binary outcomes $\pm 1$. We write the following inequality:
\begin{align}
I_{N,m} = \sum_{k=0}^{m-1} \alpha_k S_k + \frac12 \sum_{k,l} S_{kl} \geq -\beta_c \, ,
\label{eq:symBell}
\end{align}
where $\alpha_k=m-2k-1$, $\beta_c$ is the local bound, and the symmetrized correlators are defined as 
\begin{align}
S_k := \sum_{i=1}^N \langle M_k^{(i)} \rangle \, , \quad S_{kl}:= \sum_{i \neq j} \langle M_k^{(i)} M_l^{(j)}\rangle \, . 
\label{eq:correlators}
\end{align}
Let us show that~\eqref{eq:symBell} is a valid Bell inequality for $\beta_c=\left\lfloor \frac{m^2 N}{2}\right\rfloor$, where $\lfloor x\rfloor$ is the largest integer smaller or equal to $x$. Below, we assume that $m$ is even; see Appendix A for the case of odd $m$.

Since $I_{N,m}$ is linear in the probabilities and local behaviors can be decomposed as a convex combination of deterministic local strategies, the local bound of Eq.~\eqref{eq:symBell} can be reached by a deterministic local strategy~\cite{Brunner14}. We thus restrict our attention to these strategies and write
\begin{align}
\langle M_k^{(i)} \rangle = x_k^i = \pm 1 \quad \Rightarrow S_{kl}=S_k S_l -\sum_{i=1}^N x_k^i x_l^i \, ,
\label{eq:determinism}
\end{align}
where $x_k^i$ is the (deterministic) outcome party $i$ produces when asked question $k$. This directly leads to the following decomposition:
\begin{align}
I_{N,m} = \sum_{k=0}^{\frac{m}{2}-1}\alpha_k (S_k-S_{m-k-1}) +\frac12 B^2 - \frac12 C \geq -\beta_c\, ,
\label{eq:BellEven}
\end{align}
with $B:= \sum_{k=0}^{m-1}S_k$ and $C:=\sum_{i=1}^N\left(\sum_{k=0}^{m-1}x_k^i\right)^2$. Due to the symmetry under exchange of parties of this Bell expression, it is convenient to introduce, following~\cite{Tura14}, variables counting the number of parties that use a specific deterministic strategy:
\begin{align}
a_{j_1<...<j_n} :&=\#\{i \in \{1,...,N\} \vert x_k^i = -1 \text{ iff } k \in \{j_1,...,j_n\}\} \nonumber \\
\bar{a}_{j_1<...<j_n} :&=\#\{i \in \{1,...,N\} \vert x_k^i = +1 \text{ iff } k \in \{j_1,...,j_n\}\}\nonumber\\
n &\leq \frac{m}{2} \, , \quad \bar{a}_{j_1,...,j_{\frac{m}{2}}} \equiv 0 \, ,
\label{eq:strategyvar}
\end{align} 
where $\#$ denotes the set cardinality. Since each party has to choose a strategy, the variables sum up to $N$:
\begin{align}
\sum_{\text{all variables}}= \sum_{n=0}^{\frac{m}{2}}\sum_{j_1<...<j_n}\left(a_{j_1...j_n}+\bar{a}_{j_1...j_n}\right)=N \, .
\label{eq:SumN}
\end{align}
The correlators can now be expressed as
\begin{align}
S_k = \sum_{n=0}^{\frac{m}{2}}\sum_{j_1<...<j_n}\left(a_{j_1...j_n}-\bar{a}_{j_1...j_n}\right)y_k^{j_1...j_n} \, ,
\label{eq:Sky}
\end{align}
with $y_k^{j_1...j_n}= -1$ if $k\in\{j_1,...,j_n\}$, and $+1$ otherwise.

The first term of~\eqref{eq:BellEven} concerns the difference between two correlators. Let us see how this term decomposes as a function of the number of indices present in its variables. From Eq.~\eqref{eq:Sky}, it is clear that a variable with $n$ indices only appears in the difference $S_k-S_l$ if $y_k^{j_1...j_n}\neq y_l^{j_1...j_n}$. But the corresponding strategy only has $n$ differing outcomes and each correlator in this term only appears once, so a variable with $n$ indices appears in at most $n$ of these differences. Moreover, if it appears, it does so with a factor $\pm 2$. The coefficient in front of a variable with $n$ indices in the first sum of~\eqref{eq:BellEven} thus cannot be smaller than $-2\sum_{k=0}^{n-1}\alpha_k=2n(n-m)$.

The second term of~\eqref{eq:BellEven} can be bounded as $B^2\geq 0$, while the third one can be expressed as
\begin{align}
C=\sum_{n=0}^{\frac{m}{2}}\sum_{j_1<...<j_n}\left(a_{j_1...j_n}+\bar{a}_{j_1...j_n}\right) (m-2n)^2 \, .
\end{align} 
Putting everything together and using property \eqref{eq:SumN}, we arrive at 
\begin{align}
I_{N,m}&\geq \sum_{k=0}^{\frac{m}{2}-1}\alpha_k (S_k-S_{m-k-1}) - \frac12 C \nonumber\\
 &\geq -\frac{m^2}{2}\sum_{\text{all variables}}=-\frac{m^2 N}{2} = -\beta_c \, ,
\end{align}
which concludes the proof.

Note that this bound is achieved for $a_{01...\frac{m}{2}-1}=N$, i.e.\ when for each party exactly the first half of the measurements yields result $-1$. Note also that the Bell inequality~\eqref{eq:symBell} does not reduce to Ineq.~(6) of Ref.~\cite{Tura14} when $m=2$. 
Indeed, while none of these inequalities is a facet of the local polytope, the latter one is a facet of the symmetrized 2-body correlator local polytope~\cite{Tura14,symm10}.

\paragraph{From Bell inequalities to Bell-correlation witnesses --}
Let us now derive a set of Bell-correlation witnesses
assuming a certain form for the measurement operators. Here, no assumptions are made on the measured state.

Following Ref.~\cite{Schmied16}, we start from inequality \eqref{eq:symBell} and introduce spin measurements along the axes $\vec{d}_k$, $k=0,...,m-1$, as well as the collective spin observables $\hat{S}_k$: 
\begin{align}
M_k^{(i)}= \vec{d}_k \cdot \vec{\sigma}^{(i)} \, , \quad \hat{S}_k = \frac12 \sum_{i=1}^N M_k^{(i)} \, ,
\label{eq:SpinMeas}
\end{align}
where $\vec \sigma$ is the Pauli vector acting on a spin-$\frac12$ system. The correlators can be expressed in terms of these total spin observables and the measurement directions~\cite{Tura14}:
\begin{align}
S_k &= 2\langle \hat{S}_k \rangle \nonumber\\
S_{kl} &= 2\left[\left\langle\hat{S}_k\hat{S}_l\right\rangle+\left\langle\hat{S}_l\hat{S}_k\right\rangle\right] - N \vec{d}_k\cdot\vec{d}_l\, .
\end{align}
This defines the Bell operators
\begin{align}
\hat{W}_{N,m} := 2\!\sum_{k=0}^{m-1}\! \alpha_k \hat{S}_k + 2\!\sum_{k,l} \hat{S}_k\hat{S}_l -\frac{N}{2}\! \sum_{k,l} \vec{d}_k\!\cdot\!\vec{d}_l + \left\lfloor\!\frac{m^2N}{2}\!\right\rfloor \, , 
\end{align}
whose expectation values are positive for states that are not Bell-correlated. Note that the expectation value of these operators need not be negative for all Bell-correlated states and every choice of mesurement, though. A negative value may only be achieved for specific choices of states and measurement settings.

We now consider measurement directions $\vec{d}_k = \vec{a}\cos(\vartheta_k)+\vec{b}\sin(\vartheta_k)$ lying in a plane spanned by two orthonormal vectors $\vec{a}$ and $\vec{b}$, 
with the antisymmetric angle distribution $\vartheta_{m-k-1}=-\vartheta_k$. Note that the coefficients $\alpha_k$ share the same antisymmetry. Defining $\mathcal{W}_m := \left\langle \frac{\hat{W}_{N,m}}{2 \hat{N}} \right\rangle $ for even $m$, we arrive at the following family of witnesses:
\begin{align}
\mathcal{W}_{m} = \C_b {\sum_{k=0}^{\frac{m}{2}-1} \alpha_k \sin(\vartheta_k)} \,{-} {\left(1-\zeta_a^2\right)} {\left[\sum_{k=0}^{\frac{m}{2}-1}\cos(\vartheta_k)\right]^2}\! {+}\, \frac{m^2}{4} \, , 
\label{eq:witness}
\end{align}
with $\mathcal{W}_m\geq 0$ for states that are not Bell correlated. These Bell correlation witnesses depend on $\frac{m}{2}$ angles $\vartheta_k$ and involve just two quantities to be measured: the scaled collective spin $\C_b:=\left\langle \frac{\hat{S}_{\vec{b}}}{\hat{N}/2}\right\rangle$ and the scaled second moment $\zeta_a^2 :=\left\langle \frac{\hat{S}_{\vec{a}}^2}{\hat{N}/4}\right\rangle$.

The tightest constraints on $\C_b$ and $\zeta_a^2$ that allow for a violation of $\mathcal{W}_m\geq0$ are obtained by minimizing $\mathcal{W}_m$ over the angles $\vartheta_k$. Solving $\frac{\partial \mathcal{W}_m}{\partial \vartheta_k}=0$ yields (see Appendix B):
\begin{align}
\vartheta_k = -\arctan[\lambda_m (m-2k-1)] \, ,  
\label{eq:OptAng} \\
\frac{\C_b}{2\lambda_m(1-\zeta_a^2)}=\sum_{k=0}^{\frac{m}{2}-1}\cos(\vartheta_k)\, . \label{eq:Selfcons} 
\end{align} 
Equation \eqref{eq:Selfcons} is a self-consistency equation for $\lambda_m$ that has to be satisfied in order to minimize $\mathcal{W}_m$.

Using these parameters, we can rewrite our witness in terms of the physical parameters $\C_b$ and $\zeta_a^2$ only. For two measurement directions ($m=2$), we find that states which are not Bell-correlated satisfy
\begin{align}
\zeta_a^2\geq Z_2(\C_b)= \frac12\left(1-\sqrt{1-\C_b^2}\right) \, .
\label{eq:Z2}
\end{align}
This recovers the bound obtained from a different inequality in~\cite{Schmied16}. Note that in the present case, the argument is more direct since it does not involve $\C_a$, the first moment of the spin operator in the $a$ direction.

\begin{figure}
\includegraphics[width=0.9\linewidth]{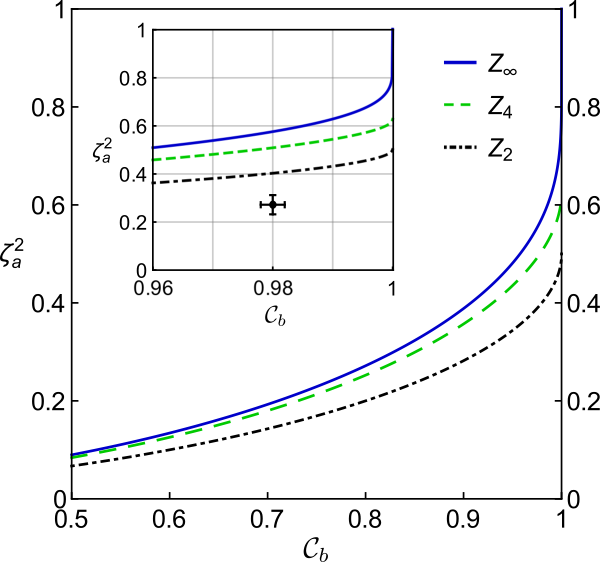}
\caption{Plots of the critical lines $Z_2$, $Z_4$ and $Z_{\infty}$. The witness obtained from the Bell inequality with $4$ settings already provides a significant improvement over the case of $2$ settings. The black point in the inset shows the data point from~\cite{Schmied16}, with $N = 476\pm 21$.}
\label{fig:Z2,4,Inf}
\end{figure}

Increasing the number of measurement directions allows for the detection of Bell correlations in additional states. In the limit $m\to\infty$, we find (see Appendix B):
\begin{align}
\zeta_a^2\geq Z_{\infty}(\C_b) = 1-\frac{\C_b}{\text{artanh}\left(\C_b\right)} \, .
\label{eq:ZInf} 
\end{align}

Figure~\ref{fig:Z2,4,Inf} shows the two witnesses~\eqref{eq:Z2} and~\eqref{eq:ZInf} together with the one obtained similarly for $m=4$ settings in the $\C_b$-$\zeta_a^2$ plane. The curve $Z_\infty$ reaches the point $\C_b=\zeta_a^2=1$, therefore allowing in principle for the detection of Bell correlations in presence of arbitrarily low squeezing. It is known, however, that some values of $\C_b$ and $\zeta_a^2$ can only be reached in the limit of a large number of spins~\cite{Sorensen01}. For any fixed $N$, a finite amount of squeezing is thus necessary in order to allow for the violation of our witness (see Appendix C). The corresponding upper bound on $\zeta_a^2$ is shown in Figure~\ref{fig:zetaUpper}.

\begin{figure}
\includegraphics[width=0.95\linewidth]{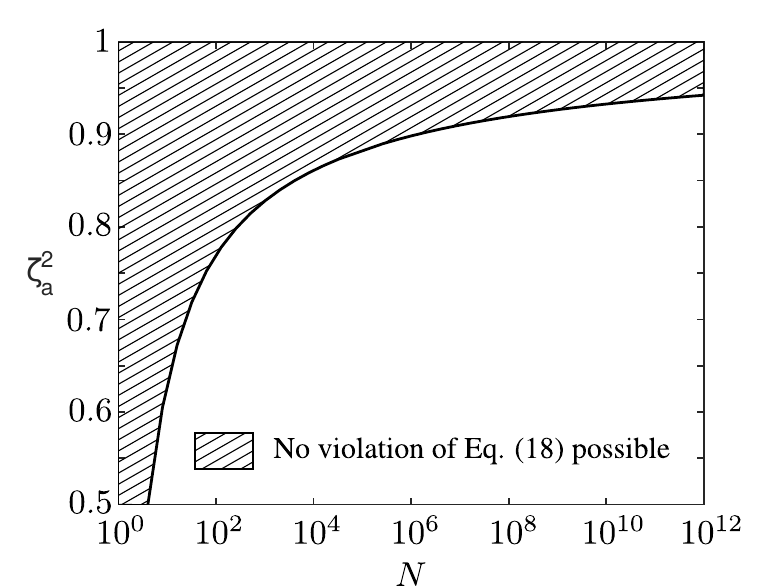}
\caption{Upper bound on the value of $\zeta_a^2$ required to see a violation of the Bell correlation witness~\eqref{eq:ZInf}. The bound depends on the number of particles $N$.}
\label{fig:zetaUpper}
\end{figure}

Points below the curve $Z_m$ in Fig.~\ref{fig:Z2,4,Inf} indicate a violation of the witness $\mathcal{W}_{m} \geq 0$ obtained from the corresponding $m$-settings Bell inequality. Violation of any such bound reveals the presence of a Bell-correlated state. However, as discussed in the introduction, conclusions in the presence of finite statistics have to be examined carefully, since in practice, one can never conclude from the violation of a witness that the measured state is Bell correlated with $100\%$ confidence. The point shown in the inset of Fig.~\ref{fig:Z2,4,Inf} corresponds to the data reported in Ref.~\cite{Schmied16} from measurements on a spin-squeezed Bose-Einstein condensate. This point clearly violates the witnesses for $m=2,4,\infty$ by several standard deviations, although the number of measurement rounds is too small to guarantee that the measured state is Bell correlated without further assumptions~\cite{Schmied16}.

\paragraph{Finite Statistics --}
In this section, we put a bound on the number of experimental runs needed to exclude with a given confidence that a measured state is not Bell-correlated. Note that such a conclusion does not follow straightforwardly from the violation of the witness by a fixed number of standard deviations. Indeed, standard deviations inform on the precision of a violation, but fail at excluding arbitrary local models~\cite{Zhang11}, including e.g.\ models which may showing non-gaussian statistics with rare events. We thus look here for a number of experimental runs that is sufficient to guarantee a p-value lower than a given threshold for the null hypothesis `The measured state is not Bell-correlated'. Since we are concerned with the characterization of physical systems in the absence of an adversary, we assume that the same state is prepared in each round (i.i.d.\! assumption).

For this statistical analysis, let us consider a different Bell correlation witness than~\eqref{eq:ZInf}. Indeed, we derived this inequality in order to maximize the amount of violation for given data, but here we rather wish to maximize the statistical evidence of a violation. For this, we take~\eqref{eq:witness} and consider the representation of the angles given in Eq. \eqref{eq:OptAng}, but without taking Eq. \eqref{eq:Selfcons} into account. In the limit of infinitely-many measurement settings, we find (see Appendix B)
\begin{align}
\mathcal{W_\text{stat}} =-\C_b \Delta_\nu &-(1-\zeta_a^2)\Lambda_\nu^2+\frac{1}{4} \geq 0 \, \text{, with} \label{eq:finalW}\\
\Delta_\nu=\frac{\sqrt{1+\nu^2}}{4\nu}-&\frac{\mathrm{arsinh}(\nu)}{4\nu^2} \, ,\quad \Lambda_\nu= \frac{\mathrm{arsinh}(\nu)}{2\nu}\, , \label{eq:DeltaLambda}
\end{align}
where $\nu=\lim\limits_{m\to\infty}\lambda_m\cdot m$ is a free parameter that fully specifies the set of measurement angles. 

In order to model the experimental evaluation of $\mathcal{W_\text{stat}}$, we introduce the following estimator:
\begin{equation}
\begin{split}
\mathcal{T}=& \frac{\chi(Z=0)}{q}X + \frac{\chi(Z=1)}{1-q}Y + (\frac14-\Delta_\nu-\Lambda_\nu^2) \, .
\end{split}
\end{equation}
Here, $\chi$ denotes the indicator function and the binary random variable $Z$ accounts for the choice between the measurement of either $\C_b$ or $\zeta_a$. Each measurement allows for the evaluation of the corresponding random variables $X=\Delta_\nu(1-\C_b)$ and $Y=\Lambda_\nu^2\zeta_a^2$. Assuming that $Z$ is independent of $X$ and $Y$ and choosing $q=P[Z=0]$ guarantees that $\mathcal{T}$ is a proper estimator of $\mathcal{W}_\text{stat}$, i.e. $\langle\mathcal{T}\rangle=\mathcal{W}$. $q$ then corresponds to the probability of performing a measurement along the $b$ axis. We choose $q=\left(1+\frac{\Lambda_\nu^2 N}{2\Delta_\nu}\right)^{-1}$ so that the contributions of both measurements to $\mathcal{T}$ have the same magnitude, i.e. the maximum values of $X/q$ and $Y/(1-q)$ are equal within the domain $|\C_b|\leq1$ and $\zeta_a^2\in[0,N]$. This also guarantees that the spectrum of $\mathcal{T}$ matches that of $\mathcal{W}_\text{stat}$.

Suppose the measured state is non-Bell-correlated, i.e. that its mean value $\mu=\langle \mathcal{T}\rangle= \mathcal{W}_\text{stat} \geq 0$. 
We are now interested in the probability that after $M$ experimental runs the estimated value $T=\frac1M \sum_{i=1}^M \mathcal{T}_i$ of the witness $\mathcal{W}_\text{stat}$ falls below a certain value $t_0<0$, with $\mathcal{T}_i$ being the value of the estimator in the $i^\text{th}$ run.

In statistics, concentration inequalities deal with exactly this issue. In Appendix D, we compare four of these inequalities, namely the Chernoff, Bernstein, Uspensky and Berry-Esseen ones and show explicitly that in the regime of interest the tightest and therefore preferred bound results from the Bernstein inequality:
\begin{align}
P[T \leq t_0]\leq \exp\left(-\frac{(\mu-t_0)^2 M}{2\sigma_0^2+\frac23 (t_u-t_l) (\mu-t_0)}\right) \leq \varepsilon\, .
\label{eq:Bern}
\end{align} 
Here, $t_0$ is the experimentally observed value of $T$ after $M$ measurement rounds, $t_l=\frac14-\Delta_\nu-\Lambda_\nu^2$ and $t_u=\frac14+\Delta_\nu+\Lambda_\nu^2(N+1)$ are lower and upper bounds on the random variable $\mathcal{T}$ respectively, and $\sigma_0^2$ is its variance for a local state. 

We show in Appendix D that the largest p-value is obtained by setting $\mu=0$ and $\sigma_0^2=-t_l t_u$.
A number of measurements sufficient to exclude the null hypothesis with a probability larger than $1-\varepsilon$ is then given by:
\begin{align}
M\geq \frac{-2t_l t_u -\frac23(t_u-t_l)t_0}{t_0^2}\ln\left(\frac1\varepsilon\right) \, . 
\label{eq:MBern}
\end{align}
This quantity can be minimized by choosing the free parameter $\nu$ appropriately.
As shown in Appendix D, optimizing $\nu$ at this stage allows us to reduce the number of measurement by $\sim\!\!30$\%. It is thus clearly advantageous not to consider the witness~\eqref{eq:ZInf} when evaluating statistical significance.

The number of runs in~\eqref{eq:MBern} depends linearly on $t_l$ and therefore also linearly on $N$. The ratio $\frac{M}{N}$ thus tends to a constant for large $N$ (see Appendix D for more details). This implies that a number of measurements growing linearly with the system size is both necessary and sufficient to close the statistics loophole~\cite{Schmied16}.

Figure \ref{fig:MNvsCb} depicts the required number of experimental runs per spin as a function of the scaled collective spin $\C_b$ and of the scaled second moment $\zeta_a^2$. For a confidence level of $1-\epsilon=99\%$, between 20 and ~500 measurement runs per spins are required in the considered parameter region.

\begin{figure}
\centering
\includegraphics[width=\linewidth]{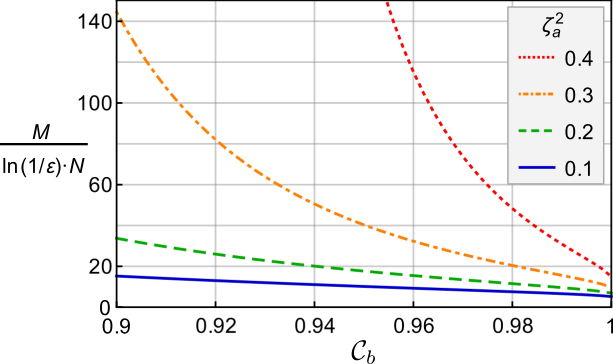}
\caption{Number of experimental runs per spins required to rule out non-Bell-correlated states with a confidence of $1-\varepsilon$ as a function of $\C_b$ and $\zeta_a$.
For $\C_b=0.98$ and $\zeta_a^2=0.272$ (as reported in \cite{Schmied16}), approximately $17\cdot\ln(100)\simeq80$ runs per spin are sufficient to reach a confidence level of $99\%$.}
\label{fig:MNvsCb}
\end{figure}

\paragraph{Conclusion --}
In this paper, we introduce a class of multipartite Bell inequalities involving two-body correlators and an arbitrary number of measurement settings. Assuming collective spin measurements, these inequalities give rise to the witness~\eqref{eq:ZInf}, which can be used to determine whether Bell correlations can be detected in a many-body system. This criterion detects states that were not detected by the previously-known witness~\cite{Schmied16}.

We then discuss the statistics loophole arising in experiments involving many-body systems, i.e. the difficulty of ruling out, without further assumptions, non-Bell-correlated states in the presence of finite statistics. We provide a bound, Eq.~\eqref{eq:MBern}, on the number of measurement rounds that allows one to close this loophole. This bound shows that all non-Bell-correlated states can be convincingly ruled out at the cost of performing a number of measurements that grows linearly with the system size. This opens the way for a demonstration of Bell-correlations in a many-body system free of the statistics loophole.

\paragraph{Acknowledgements --} We thank Baptiste Allard, Remik Augusiak and Valerio Scarani for helpful discussions. This work was supported by the Swiss National Science Foundation (SNSF) through grants PP00P2-150579, 20020-169591 and NCCR QSIT. NS acknowledges the Army Research Laboratory Center for Distributed Quantum Information via the project SciNet.

\onecolumngrid
\appendix

\section{Proof of the Bell inequalities}

In this appendix, we expand on the proof of Ineq.~\eqref{eq:symBell} given in the main text, and cover the case of odd numbers of measurements.

\subsection{Symmetric Bell inequality for $m$ measurements}
We consider local measurements on $N$ parties. For each party, one can choose between $m$ measurements $M_k^{(i)}$, where $k \in \{0,1,...,m-1\}$ and $i \in \{1,...,N\}$. Each measurement has the two possible outcomes $\pm 1$. We are interested in Bell inequalities, i.e. inequalities every local theory has to obey \cite{Brunner14}. We only consider one- and two-body mean values, so the general form of such an inequality is
\begin{align}
I_{N,m} &= \sum_{k=0}^{m-1}\sum_{i=1}^N \alpha_k^i \langle M_k^{(i)} \rangle + \sum_{k,l} \sum_{i < j} \beta_{kl}^{ij} \langle M_k^{(i)} M_l^{(j)} \rangle \geq -\beta_c \, ,
\label{eq:generalI}
\end{align}
where $\langle M_k^{(i)} \rangle=\sum_{a\in\{-1,1\}}a\, \text{Prob}(M_k^{(i)}=a)$ and $\langle M_k^{(i)} M_l^{(j)} \rangle=\sum_{a,b\in\{-1,1\}}ab\, \text{Prob}(M_k^{(i)}=a,M_l^{(j)}=b)$.

We now restrict ourselves to Bell inequalities which are symmetric under exchange of parties, i.e. $\alpha_k^i=\alpha_k$ and $\beta_{kl}^{ij}=\beta_{kl}$. After defining the symmetrized correlators 
\begin{align}
S_k = \sum_{i=1}^N \langle M_k^{(i)} \rangle \quad , \quad S_{kl} = \sum_{i\neq j} \langle M_k^{(i)}M_l^{(j)}\rangle \, ,
\label{eq:Cor}
\end{align}
symmetric inequalities can be expressed as
\begin{align}
I_{N,m} = \sum_{k=0}^{m-1}\alpha_k S_k +\frac{1}{2}\sum_{k,l}^{m-1}\beta_{kl} S_{kl} \geq -\beta_c \, .
\end{align}

We are interested in cases for which the coefficients are $\alpha_k=m-2k-1$ ($ k=0,...,m-1$) and $\beta_{k,l}=1$. We note that $\alpha_{m-k-1}=-\alpha_k$ and claim that local theories have to fulfill the Bell inequalities
\begin{align}
I_{N,m} &= \sum_{k=0}^{m-1}(m-2k-1)S_k +\frac{1}{2}\sum_{k,l}^{m-1}S_{kl} \geq -\left\lfloor\frac{m^2N}{2}\right\rfloor = -\beta_c \, ,
\label{eq:final_I}
\end{align}
where $\lfloor x\rfloor$ is the largest integer smaller or equal to $x$.

\subsection{Computation of the local bound}

In this section we prove the claim above. One of the most important properties of a local theory is its equivalence to a mixture of deterministic local theory. That is why, by considering only deterministic theories, there is no loss of generality. We can therefore assume that a measurement $M_k^{(i)}$ will lead to an outcome $x_k^i=\pm 1$ with probability $1$, i.e. $\langle M_k^{(i)} \rangle = x_k^i$. The two-body correlators $S_{kl}$ can thus be expressed as $S_{kl}=S_k S_l-\sum_{i=1}^N x_k^i x_l^i$. By also taking the antisymmetry of $\alpha_k$ into account, and introducing the quantities
\begin{align}
A=\sum_{k=0}^{\left\lfloor\frac{m}{2}\right\rfloor -1}(m-2k-1) (S_k-S_{m-k-1})\ ,\quad
B=\sum_{k=0}^{m-1} S_k\ ,\quad
C=\sum_{i=1}^N \left[\sum_{k=0}^{m-1}x_k^i\right]^2 \, .
\label{eq:detI}
\end{align}
we arrive at
\begin{align}
I_{N,m} = A + \frac12 B^2 - \frac12 C\, .
\end{align}

\subsubsection{Strategy variables}
We want to rewrite $I_{N,m}$ further. Therefore we introduce variables counting the strategies chosen by the parties. 
Because there are $m$ measurements with binary outcomes, the number of possible strategies per party is $2^m$. We define the following $2^m$ variables: 
\begin{align}
a_{j_1<...<j_n} :&=\#\{i \in \{1,...,N\} \vert x_k^i = -1 \text{ iff } k \in \{j_1,...,j_n\}\} \quad\text{ for } n\leq  \left\lfloor\frac{m}{2}\right\rfloor \, , \nonumber \\
\bar{a}_{j_1<...<j_n} :&=\#\{i \in \{1,...,N\} \vert x_k^i = +1 \text{ iff } k \in \{j_1,...,j_n\}\} \quad \text{ for } n\leq  \left\lfloor\!\frac{m-1}{2}\!\right\rfloor \, , \nonumber\\
\bar{a}_{j_1,...,j_{\left\lceil\!\frac{m}{2}\!\right\rceil}} &\equiv 0 \, ,
\label{eq:Strategy}
\end{align}
where $\#$ denotes the set cardinality. For example, $a_j$ counts the parties $k$ whose outcomes are $x_k^{j'}=1-2\delta_{jj'}$. $\bar{a}_j$ is the number of parties following the opposite strategy. Variables with $n$ indices thus correspond to a strategy for which either exactly $n$ of the $m$ outcomes are $+1$ or exactly $n$ of the outcomes are $-1$, i.e. $n$ outcomes differ from the rest. Note that the conjugate variables in the case of $\frac{m}{2}$ indices are set to zero for the case of even $m$ in order to prevent strategies from being counted twice. Since every party has to choose one strategy, the variables sum up to $N$, i.e.
\begin{align}
a+\bar{a}+\sum_{j=0}^{m-1}(a_j+\bar{a}_j) + ... =\sum_{n=0}^{\left\lfloor\frac{m}{2}\right\rfloor}\sum_{j_1<...<j_n} (a_{j_1...j_n}+\bar{a}_{j_1...j_n})=N \, \text{.} 
\label{eq:SumStrategies}
\end{align} 

Note that $S_k$ can be expressed in terms of the strategy variables as follows:
\begin{align}
&S_k = \sum_{n=0}^{\left\lfloor\frac{m}{2}\right\rfloor}\sum_{j_1<...<j_n}\left(a_{j_1...j_n}-\bar{a}_{j_1...j_n}\right)y_k^{j_1...j_n} \text{ , with} \label{eq:S_k}\\
&y_k^{j_1...j_n} =\begin{cases}-1 \quad \text{if }k\in\{j_1,...,j_n\} \\+1 \quad \text{else}\end{cases}  .\label{eq:outcomes}
\end{align}

\subsubsection{Decomposition of $A$ and $C$ in terms of the strategy variables}
Equation \eqref{eq:S_k} results in the following representation of $S_k-S_l$:
\begin{align}
S_k-S_l = \sum_{n=0}^{\left\lfloor\frac{m}{2}\right\rfloor}\sum_{j_1<...<j_n}\left(a_{j_1...j_n}-\bar{a}_{j_1...j_n}\right)\left(y_k^{j_1...j_n}-y_l^{j_1...j_n}\right) \text{ .}
\end{align}
Clearly, a variable only appears in this expression if $y_k \neq y_l$, i.e. if the strategy is such that the outcome of the $k^{th}$ measurement differs from the $l^{th}$. This, for example, cannot be the case if the number of indices is zero, i.e. if all measurement outcomes are the same. So $a$ and $\bar{a}$ do not show up in $S_k-S_l$. 

With the help of the introduced strategy variables, we can express $A$ as
\begin{align}
A=\sum_{k=0}^{\left\lfloor\frac{m}{2}\right\rfloor -1}(m-2k-1) \sum_{n=0}^{\left\lfloor\frac{m}{2}\right\rfloor}\sum_{j_1<...<j_n}\left(a_{j_1...j_n}-\bar{a}_{j_1...j_n}\right)\left(y_k^{j_1...j_n}-y_{m-k-1}^{j_1...j_n}\right) \text{ .}
\end{align}
and $C$ as
\begin{align}
C &= m^2(a+\bar{a}) + (m-2)^2\sum_{j=0}^{m-1}(a_j+\bar{a}_j)+(m-4)^2\sum_{j_1<j_2}(a_{j_1j_2}+\bar{a}_{j_1j_2})+... \nonumber \\
  &= \sum_{n=0}^{\left\lfloor\frac{m}{2}\right\rfloor}(m-2n)^2\sum_{j_1<...<j_n}(a_{j_1...j_n}+\bar{a}_{j_1...j_n}) \, .\label{eq:C}
\end{align}
In other words, we notice that if a variable has $n$ indices it contributes to $C$ with a factor $(m-2n)^2$.

\subsubsection{A bound independent of the number of indices}
In this section, we study the contributions of $A$ and $C$ to $I_{N,m}$. 
For this, we make use of the following theorem:

\begin{TH1}
A strategy with $n$ equal outcomes satisfies the inequality
\begin{align*}
\sum_{k=0}^{\left\lfloor\frac{m}{2}\right\rfloor-1}\left\vert y_k^{j_1...j_n}-y_{m-k-1}^{j_1...j_n}\right\vert \leq 2n \text{ .}
\end{align*}
\label{Th1}
\end{TH1}
\begin{proof}
First we note that the summation is such that no $y_k^{j_1...j_n}$ appears twice. Also, we know that since we consider binary outcomes, $\left\vert y_k-y_{m-k-1}\right\vert$ is either $0$ or $2$. We thus have
\begin{align*}
\sum_{k=0}^{\left\lfloor\frac{m}{2}\right\rfloor-1}\left\vert y_k^{j_1...j_n}-y_{m-k-1}^{j_1...j_n}\right\vert = 2l \, ,\quad l \in \mathbb{N} \, .
\end{align*}
Assume now that the above inequality is violated, i.e. $l>n$.

$\Leftrightarrow y_k^{j_1...j_n} \neq y_{m-k-1}^{j_1...j_n}$ for $l>n$ values of $k$.

$\Leftrightarrow$ The strategy $(j_1,...,j_n)$ has $l$ differing outcomes.

This is a contradiction to the definition of the strategy. Therefore the assumption must be wrong and the inequality holds for all strategies.  
\end{proof}
\begin{COR}
A function $f(k)$ which is monotonically decreasing with $k$, satisfies
\begin{align*}
\sum_{k=0}^{\left\lfloor\frac{m}{2}\right\rfloor-1} f(k)\left\vert y_k^{j_1...j_n}-y_{m-k-1}^{j_1...j_n}\right\vert \leq 2\sum_{k=0}^{n-1} f(k) \, .
\end{align*} 
\label{Co1}
\end{COR}
\begin{proof}
Theorem \ref{Th1} implies that $y_k \neq y_{m-k-1}$ for at most $n$ values of $k$. Taking into account that $f(k)$ is monotonically decreasing, we find that 
\begin{align*}
\sum_{k=0}^{\left\lfloor\frac{m}{2}\right\rfloor-1} f(k)\left\vert y_k^{j_1...j_n}-y_{m-k-1}^{j_1...j_n}\right\vert \leq \sum_{\substack{\text{n values} \\ \text{of k}}} f(k)\cdot 2 \leq 2 \sum_{k=0}^{n-1} f(k) \, .
\end{align*} 
\end{proof}

With the help of Corollary~\ref{Co1}, we rewrite the quantity $A$ as
\begin{align}
A &=\sum_{n=0}^{\left\lfloor\frac{m}{2}\right\rfloor}\sum_{j_1<...<j_n}\left(a_{j_1...j_n}-\bar{a}_{j_1...j_n}\right)\sum_{k=0}^{\left\lfloor\frac{m}{2}\right\rfloor-1}\alpha_k \left(y_k^{j_1...j_n}-y_{m-k-1}^{j_1...j_n}\right)\nonumber\\
&\geq \sum_{n=0}^{\left\lfloor\frac{m}{2}\right\rfloor}\sum_{j_1<...<j_n}\left(a_{j_1...j_n}-\bar{a}_{j_1...j_n}\right)(-2)\sum_{k=0}^{n-1} \alpha_k 
= -2\sum_{n=0}^{\left\lfloor\frac{m}{2}\right\rfloor}n(m-n)\sum_{j_1<...<j_n}\left(a_{j_1...j_n}-\bar{a}_{j_1...j_n}\right) \, .
\end{align}
Making use of Eq.~\eqref{eq:C} and \eqref{eq:SumN}, we then find that 
\begin{align}
A-\frac12C &\geq \sum_{n=0}^{\left\lfloor\frac{m}{2}\right\rfloor}\left[-2n(m-n)-\frac{(m-2n)^2}{2}\right]\sum_{j_1<...<j_n}\left(a_{j_1...j_n}-\bar{a}_{j_1...j_n}\right)  \nonumber\\
&= -\frac{m^2}{2}\sum_{n=0}^{\left\lfloor\frac{m}{2}\right\rfloor}\sum_{j_1<...<j_n}\left(a_{j_1...j_n}-\bar{a}_{j_1...j_n}\right) = -\frac{m^2}{2}N
\end{align}

\subsubsection{Putting the pieces together}
In order to conclude the proof, we now only miss the contribution of the term $B$. For this, we look at the case of even and odd $m$ separately.
When $m$ is even, $B=\sum_k S_k$ is also even. We thus find that $B^2 \geq 0$. This means that 
\begin{align}
I_{N,m} \geq A-\frac12 C \geq -\frac{m^2 N}{2} \, .
\end{align}
If $m$ is odd, $B$ shares the parity of $N$. That is why we have $B^2\geq 0$ for even $N$, and  $B^2\geq 1$ for odd $N$,  resulting in 
\begin{align}
I_{N,m} \geq \begin{cases} -\frac{m^2 N}{2} &\text{ for even N}\\ -\frac{m^2 N}{2}+\frac{1}{2} \quad &\text{ for odd N} \end{cases} \text{ .}
\end{align}

In general, the classical bound is thus $\beta_c=\left\lfloor \frac{m^2 N}{2} \right\rfloor$.

\section{Optimization of the witnesses}
In this appendix, we optimize the witnesses $\mathcal{W}_m$ as given in Eq.~\eqref{eq:witness} of the main text over the measurement angles.  
Let us remind the form of $\mathcal{W}_m$:
\begin{align}
\mathcal{W}_{m} = \C_b {\sum_{k=0}^{\frac{m}{2}-1} \alpha_k \sin(\vartheta_k)} \,{-} {\left(1-\zeta_a^2\right)} {\left[\sum_{k=0}^{\frac{m}{2}-1}\cos(\vartheta_k)\right]^2}\! {+}\, \frac{m^2}{4} \, .
\label{eq:witness.A}
\end{align}
We do this optimization by searching for those angles leading to the minimum of $\mathcal{W}_m$. This is equivalent to solving the system of equations arising from $\frac{\partial \mathcal{W}_m}{\partial \vartheta_k}=0$:
\begin{align}
\frac{\partial\mathcal{W}_m}{\partial\vartheta_k}=(m-2k-1)\C_b\cos(\vartheta_k)+2\sin(\vartheta_k)\left(1-\zeta_a^2\right)\sum_{l=0}^{\frac{m}{2}-1}\cos(\vartheta_l) = 0 \, . \label{eq:Partial1}
\end{align}
We eventually want to find angles such that $\mathcal{W}_m$ is negative. To achieve this, the last term of Eq. \eqref{eq:witness.A} must be compensated. Since $\zeta_a^2 \geq 0$, the second term of Eq. \eqref{eq:witness.A} is bounded by $-\frac{m^2}{4}$ and thus cannot be sufficient for a negative $\mathcal{W}_m$. On the other hand, the first term is bounded by $-\frac{m^2}{4}+1$ due to the fact that $\vert \C_b\vert\leq 1$. So we find that in order to reach $\mathcal{W}_m<0$, we need $\sin(\vartheta_k)$, $\cos(\vartheta_k)$ and $\C_b$ to differ from zero and $\zeta_a^2<1$. In the following studies, we assume these necessary constraints, allowing us to rewrite Eq.~\eqref{eq:Partial1} as
\begin{align}
\frac{2(1-\zeta_a^2)}{\C_b}\sum_{l=0}^{\frac{m}{2}-1}\cos(\vartheta_l) = -(m-2k-1)\frac{\cos(\vartheta_k)}{\sin(\vartheta_k)}  \quad \forall k \, . \label{eq:Partial2}
\end{align}
Since the left side of Eq. \eqref{eq:Partial2} does not explicitly depend on $k$, this can only be achieved if both sides are equal to a constant. The assumptions about $\zeta_a^2$, $\C_b$ and the angles, as reasoned above, allow us to write
\begin{align}
\frac{2(1-\zeta_a^2)}{\C_b}\sum_{l=0}^{\frac{m}{2}-1}\cos(\vartheta_l)=\frac{1}{\lambda_m} = -(m-2k-1)\frac{\cos(\vartheta_k)}{\sin(\vartheta_k)}
\, , \label{eq:Partial3}
\end{align}
where $\lambda_m$ is a constant depending for given $\C_b$ and $\zeta_a^2$ only on $m$. We find the optimal angles
\begin{align}
\vartheta_k = -\arctan\left(\lambda_m(m-2k-1)\right) \, .\label{eq:atan} 
\end{align}
For a minimal $\mathcal{W}_m$, the constants $\lambda_m$ have to fulfill the self-consistency equations  
\begin{align}
\frac{\C_b}{2\lambda_m(1-\zeta_a^2)}=\sum_{l=0}^{\frac{m}{2}-1}\cos(\vartheta_l)=\sum_{l=0}^{\frac{m}{2}-1}\frac{1}{\sqrt{1+\lambda_m^2(m-2l-1)^2}}\, . \label{eq:selfc}
\end{align}
Here, we used the fact that $\cos(\arctan(x))=\frac{1}{\sqrt{1+x^2}}$.
For further steps, we note that $\sin(\arctan(x))=\frac{x}{\sqrt{1+x^2}}$ and define the following functions:
\begin{align}
\Lambda_m(\lambda_m) :&= \sum_{k=0}^{\frac{m}{2}-1}\frac{1}{\sqrt{1+\lambda_m^2(m-2k-1)^2}}
= \,\sum_{k=1}^{\frac{m}{2}}\frac{1}{\sqrt{1+\lambda_m^2(2k-1)^2}} , \label{eq:Lambda.A} \\
\Delta_m(\lambda_m)  :&= \sum_{k=0}^{\frac{m}{2}-1}\frac{\lambda_m(m-2k-1)^2}{\sqrt{1+\lambda_m^2(m-2k-1)^2}} 
=\sum_{k=1}^{\frac{m}{2}}\frac{\lambda_m(2k-1)^2}{\sqrt{1+\lambda_m^2(2k-1)^2}} \label{eq:Delta.A} \, .
\end{align} 
If we assume the representation of the angles given in Eq. \eqref{eq:atan}, the witnesses can be expressed as
\begin{align}
\mathcal{W}_m = - \C_b\Delta_m(\lambda_m)-\left(1-\zeta_a^2\right)\Lambda_m^2(\lambda_m) +\frac{m^2}{4} \geq 0\, , \label{eq:finalW.A} 
\end{align}
which holds for non-Bell-correlated states. Note that in this expression, we only assume the $arctan$-angle-distribution, without taking the self-consistency equations into account, i.e. without optimizing the actual differences between angles.

For the case of $m\to\infty$, we need to rewrite the above witnesses, since $\mathcal{W}_m$ diverges in this limit. We define
\begin{align}
\mathcal{W}_m':= \frac{\mathcal{W}_m}{m^2} = -\C_b\frac{\Delta_m(\lambda_m)}{m^2}-\left(1-\zeta_a^2\right)\left(\frac{\Lambda_m(\lambda_m)}{m}\right)^2 +\frac{1}{4} \geq 0 \, .
\label{eq:witness_inf}
\end{align}
We also have to rewrite the constant $\lambda_m$. We define $\lambda_m=\frac{\nu_m}{m}$ and rewrite Eq. \eqref{eq:selfc} as
\begin{align}
\frac{\C_b}{2\nu_m(1-\zeta_a^2)}=\frac{\Lambda_m\left(\frac{\nu_m}{m}\right)}{m} \, . \label{eq:selfc_nu}
\end{align}
If we define $s_k=\frac{2k-1}{m}$, we see that $\frac1m$ can be expressed as $\frac{s_{k+1}-s_k}{2}$. Note that for $m\to\infty$, $s_1\to 0$ and $s_{m/2}\to 1$. Using the convention $\nu_\infty=\nu$, we find
\begin{align}
\Lambda_\nu :=\lim_{m \to \infty} \frac{\Lambda_m(\nu_m/m)}{m} &= \lim_{m \to \infty} \sum_{k=1}^{\frac{m}{2}}\frac{1}{\sqrt{1+\nu_m^2\frac{(2k-1)^2}{m^2}}} \cdot\frac{1}{m} =\lim_{m \to \infty}\sum_{k=1}^{\frac{m}{2}}\frac{1}{\sqrt{1+\nu_m^2 s_k^2}}\frac{s_{k+1}-s_k}{2} \nonumber\\
		&= \frac12\int_0^1 \frac{1}{\sqrt{1+\nu^2 s^2}}ds =\frac{\mathrm{arsinh}(\nu)}{2\nu} \label{eq:LambdaInf}\\
\Delta_\nu:= \lim_{m \to \infty} \frac{\Delta_m(\nu_m/m)}{m^2} &= \lim_{m \to \infty}\sum_{k=1}^{\frac{m}{2}}\frac{\nu_m \frac{(2k-1)^2}{m^2}}{\sqrt{1+\nu_m^2\frac{(2k-1)^2}{m^2}}}\cdot\frac{1}{m} = \lim_{m \to \infty}\sum_{k=1}^{\frac{m}{2}} \frac{\nu_m s_k^2}{\sqrt{1+\nu_m^2 s_k^2}}\frac{s_{k+1}-s_k}2 \nonumber \\
&= \frac12 \int_0^1 \frac{\nu s^2}{\sqrt{1+\nu^2s^2}}ds = \frac{\sqrt{1+\nu^2}}{4\nu}-\frac{\mathrm{arsinh}(\nu)}{4\nu^2} \label{eq:DeltaInf} \, .
\end{align}
This yields the witness
\begin{align}
\mathcal{W}_\infty'&=-\C_b\Delta_\nu - (1-\zeta_a^2)\Lambda_\nu+\frac14\\
&= -\C_b\left(\frac{\sqrt{1+\nu^2}}{4\nu}-\frac{\mathrm{arsinh}(\nu)}{4\nu^2}\right) -(1-\zeta_a^2)\frac{\mathrm{arsinh}^2(\nu)}{4\nu^2} +\frac14 \geq 0 \, .
\label{eq:WInf}
\end{align}   

We now search for those points in the $\C_b$-$\zeta_a^2$-plane that allow for a violation of the correlation witnesses of Ineq.~\eqref{eq:finalW.A} and \eqref{eq:WInf}. For this purpose, we assume the optimal angles given in Eq.~\eqref{eq:selfc} and \eqref{eq:selfc_nu} respectively. 
We define $Z_m(\C_b)$ to be the scaled second moment, as a function of the scaled collective spin, such that $\mathcal{W}_m$ vanishes.
 
\subsubsection{$m=2$}
In the case of $m=2$ measurement settings, we have to solve the following system of equations in order to find $Z_2$:
\begin{align}
0 =\mathcal{W}_2 &= -\C_b \frac{\lambda_2}{\sqrt{1+\lambda_2^2}}-(1-\zeta_a^2)\frac{1}{1+\lambda_2^2}+1 \, ,\\
\frac{\C_b}{2\lambda_2(1-\zeta_a^2)} &= \frac{1}{\sqrt{1+\lambda_2^2}} \, .
\end{align} 
From these, we find the critical line $Z_2$ and therefore the following condition, satisfied by every non-Bell-correlated state:
\begin{align}
\zeta_a^2 \geq Z_2(\C_b) = \frac12\left(1-\sqrt{1-\C_b^2}\right) \, . \label{eq:Z2.A}
\end{align}  

\subsubsection{Limit $m\to \infty$ }
The critical line $Z_\infty$ is determined by solving the following equation for $\zeta_a^2$
\begin{align}
 0 = \mathcal{W}_\infty' = -\C_b\left(\frac{\sqrt{1+\nu^2}}{4\nu}-\frac{\mathrm{arsinh}(\nu)}{4\nu^2}\right) -(1-\zeta_a^2)\frac{\mathrm{arsinh}^2(\nu)}{4\nu^2} +\frac14 \, , \quad\text{where } \nu=\sinh\left(\frac{\C_b}{1-\zeta_a^2}\right) \, .
 \label{eq:bestNu}
\end{align}
We find that any non-Bell-correlated state satisfies 
\begin{align}
\zeta_a^2\geq Z_\infty(\C_b)=1-\frac{\C_b}{\mathrm{artanh}(\C_b)} \, . \label{eq:ZInf.A}
\end{align}

\section{Squeezing requirement}

Here, we find a bound on the amount of squeezing that is needed as a function of the number of spins $N$ in order to violate the Bell correlation witness~\eqref{eq:ZInf} described in the main text.
Due to the structure of spin systems, the first moment $\C_b$, the second moment $\zeta_a^2$ and the number of spins $N$ satisfy the following constraints~\cite{Sorensen01}:
\begin{equation}
\zeta_a^2 \geq 1-\frac{N}{2}\left[\sqrt{(1-\C_b^2)\left[\left(1+\frac{2}{N}\right)^2 - \C_b^2\right]} + \C_b^2 - 1\right]
\end{equation}
(notice that there is an error in the expression given in the reference). For any number of spins $N$, equating the right-hand side of this constraint with the right-hand side of Eq.~\eqref{eq:ZInf} gives the maximum value of $\C_b$ under which a violation of the witness~\eqref{eq:ZInf} is possible. The corresponding maximum value of $\zeta_a^2$ is plotted as a function of the number of spins $N$ in Fig.~\ref{fig:zetaUpper}. For large $N$, this function can be expanded as
\begin{equation}
Z_N^* = 1-\frac{1}{\omega}-\frac1{2\omega^3}-\frac3{4\omega^4}-O(\omega^{-5})
\end{equation}
where $\omega=W_{-1}\left(-\frac{1}{2\sqrt{N+1}}\right)$ and $W_{-1}$ is the lower branch of the Lambert $W$ function.

\section{Finite statistics}
In this appendix, we introduce four concentration inequalities and determine their bound on the p-value for generic non-Bell-correlated states. We also compare these p-values and choose the optimal one to estimate a number of experimental runs sufficient to exclude non-Bell-correlated states with a confidence $1-\varepsilon$. Eventually we minimize this number of runs by optimizing the measurement angles.
 
\subsection{Concentration inequalities}
In statistics, concentration inequalities bound the probability that a random variable $X$ exceeds or falls below a certain value.
In what follows, we recall the definitions of some of these inequalities. We then discuss some of their properties in view of our problem in the following section.

\subsubsection{Chernoff bound}
The following version of the Chernoff bound was proven by Van Vu at the University of California, San Diego \cite{Chernoff}. It was done for discrete, independent random variables. However, the bound also applies in the case of continuous random variables. 
\begin{TH1}
Let $X_1,...,X_M$ be independent random variables with $\vert X_i \vert \leq 1$ and expectation values $E[X_i]=0$ for all $i$. Let $X=\sum\limits_{i=1}^M X_i$ and $\sigma^2$ be the variance of $X$. Then 
\begin{align*}
&P[X \leq - \lambda \sigma] \leq \exp\left(-\frac{\lambda^2}{4}\right) \, , &\text{for } 0 \leq \lambda \leq 2\sigma \, ,\quad\\
&P[X\leq -x_0] \leq  \exp\left(-\frac{x_0^2}{4\sigma^2}\right)  \, ,  &\text{for } 0 \leq x_0 \leq 2\sigma^2 \, .\,
\end{align*}
\label{th:Chern}
\end{TH1}
\begin{COR}
Let $X_1,...,X_M$ be independent random variables with $E[X_i]=\mu$ and $a\leq X_i\leq b$ for all $i$. Let $X=\frac1M \sum\limits_{i=1}^M X_i$, $\sigma_i^2=\mathrm{Var}[X_i]$ and $\sigma_0^2=\max\limits_i\left\{\sigma_i^2\right\}$. Then 
\begin{align*}
&P\left[X\leq x_0\right] \leq \exp\left( -\frac{(\mu-x_0)^2 M}{4 \sigma_0^2}\right) \, , &\text{for } \mu \geq x_0 \geq \mu-\frac{\sum_i \sigma_i^2}{M(b-a)} \, .
\end{align*}
\label{cor:Chern}
\end{COR}

\subsubsection{Bernstein inequality}
The following expression known as Bernstein inequality, was proven by Bernstein in 1927, but we refer to the work of George Bennett \cite{Bernstein62}.
The inequality
is valid under certain restrictions for the absolute moments. Since our random variables are bounded, we can be sure these restrictions to be fulfilled. 
\begin{TH1}
Let $X_1,...,X_M$ be independent random variables with $E[X_i]=0$ and $\vert X_i \vert \leq \xi$ for all $i$. Also let $X=\frac{1}{M}\sum\limits_{i=1}^M X_i$ and $\sigma^2=\frac{1}{M}\sum\limits_{i=1}^M \mathrm{Var}[X_i]$. Then
\begin{align*}
&P[X \geq x_0] \leq \exp\left(-\frac{x_0^2 M}{2\sigma^2+\frac23 \xi x_0}\right)\, , &\forall x_0> 0 \, , \\
&P[X \leq -x_0] \leq \exp\left(-\frac{x_0^2 M}{2\sigma^2+\frac23 \xi x_0}\right)\, , &\forall x_0> 0 \, .
\end{align*}
\label{th:bern}
\end{TH1}
\begin{COR}
Let $X_1,...,X_M$ be independent random variables with $E[X_i]=\mu$, $\sigma_i^2=\mathrm{Var}[X_i]$ and $a\leq X_i \leq b$ for all $i$. Also let $X=\frac{1}{M}\sum\limits_{i=1}^M X_i$ and $\sigma^2=\frac{1}{M}\sum\limits_{i=1}^M \sigma_i^2\leq \sigma_0^2 = \max\limits_i\{\sigma_i^2\}$. Then 
\begin{align*}
P[X \leq x_0] &\leq \exp\left(-\frac{(\mu-x_0)^2 M}{2\sigma^2+\frac23 (b-a) (\mu-x_0)}\right) \\
			  &\leq \exp\left(-\frac{(\mu-x_0)^2 M}{2\sigma_0^2+\frac23 (b-a) (\mu-x_0)}\right)\, , &\forall x_0< \mu \, .
\end{align*}
\label{cor:Bern}
\end{COR}

\subsubsection{Uspensky inequality}
Uspensky stated in \cite{Usp} an inequality for a stochastic variable:
\begin{TH1}
Let $X$ be a random variable with mean value $E[X]=0$ and $a\leq X\leq b$. Additionally $b\geq \vert a \vert$. Let $\sigma^2=\mathrm{Var}[X]$. Then for $x_0\leq 0$ 
\begin{align*}
P[X\leq x_0] \leq \frac{\sigma^2}{\sigma^2+x_0^2} \, . 
\end{align*}
\end{TH1} 
\begin{COR}
Let $X_1,...,X_M$ be independent random variables with mean values $E[X_i]=\mu$ and $a\leq X_i\leq b$ for all $i$. Additionally $b\geq \vert a \vert$. Let $\sigma_i^2= \mathrm{Var}[X_i]$ and $\sigma_0^2=\max\limits_i\left\{\sigma_i^2\right\}$. Then $\sigma^2=\mathrm{Var}[X]=\mathrm{Var}\left[\frac1M \sum\limits_{i=1}^M X_i\right] = \frac{1}{M^2}\sum\limits_{i=1}^M \sigma_i^2\leq \frac1M \sigma_0^2$ so that for $x_0\leq \mu$ 
\begin{align*}
P[X\leq x_0]\leq \frac{\sigma_0^2}{\sigma_0^2+(x_0-\mu)^2 M} \, .
\end{align*} 
\label{cor:Usp}
\end{COR}

\subsubsection{Berry-Esseen inequality}
Andrew C. Berry and Carl-Gustav Esseen proved the following theorem \cite{Berry}:
\begin{TH1}
Let $X_1,...,X_M$ be independent identically distributed random variables with $E[X_i]=\mu$ for all $i$ and with $X=\frac1M \sum\limits_{i=1}^M X_i$. Additionally, the variance $\sigma^2=\mathrm{Var}[X_i]=E[(X_i-\mu)^2]$ and the third absolute moment $\rho=E\left[\vert X_i-\mu\vert^3\right]$ are finite. Then there exists a constant $C$ such that for all $x$
\begin{align}
\vert F_M(x)-\Phi(x)\vert \leq \frac{C \rho}{\sigma^3 \sqrt{M}} \, , 
\end{align}
where $F_M(x)=P\left[\sqrt{\frac{M}{\sigma^2}}(X-\mu)\leq x\right]$ and $\Phi(x)=\frac12\left(1+\mathrm{erf}(x)\right)$.

Esseen also proved that the constant $C$ has to fulfill $C\geq\frac{\sqrt{10}+3}{6\sqrt{2\pi}}$.
A good estimate for $C$ follows from
\begin{align}
\vert F_M(x)-\Phi(x)\vert \leq \frac{0.33554(\rho+0.415\sigma^3)}{\sigma^3 \sqrt{M}} \, . 
\end{align}
\label{th:berry}
\end{TH1} 
\begin{COR}
A direct consequence of Theorem \ref{th:berry} is 
\begin{align*}
P\left[X\leq x_0\right]&=P\left[\sqrt{\frac{M}{\sigma^2}}(X-\mu)\leq \sqrt{\frac{M}{\sigma^2}}(x_0-\mu)\right] \\
	&\leq \Phi\left(\sqrt{\frac{M}{\sigma^2}}(x_0-\mu)\right) + \frac{C\rho}{\sigma^3 \sqrt{M}} \, .
\end{align*}
\label{cor:Berry}
\end{COR}

\subsection{Largest p-value of the concentration inequalities}

In this section we determine the largest p-value for the concentration inequalities listed above, under the assumption that $x_0<0$ and $\mu\geq 0$. We denote with $X$ the random variable $\frac1M\sum_{i=1}^M X_i$, with $x_l\leq X_i \leq x_u$ and $x_l<0$. 
 
The crucial point here is that we have no information about the random variable apart from its non-negative mean value. This lack of information directly disqualifies the Berry-Esseen inequality of Corollary \ref{cor:Berry} as a potential tight bound. Indeed, the Berry-Esseen bound does not result in a tighter restriction than the trivial bound $P[X\leq x_0] \leq 1$. To see this, we consider the following distribution with three peaks:

Consider the case, for which all $X_i$ satisfy the following probability distribution:
\begin{align}
P[X_i=x]=
\begin{cases}
p_l \quad &\text{ for }x=x_l \\
p_u \quad &\text{ for }x=x_u \\
1-p_l-p_u \quad &\text{ for }x=0 \\
0 \quad &\text{ else} 
\end{cases} \, ,
\end{align}
where $p_l, p_u<1$. Additionally, we demand $E[X_i]=0$ leading to the condition $-p_l x_l =p_u x_u$. We thus arrive at the variance and third absolute moment:
\begin{align}
\sigma^2 &= p_u x_u(x_u-x_l)\, , \\
\rho     &= p_u x_u(x_u^2+x_l^2) \\
\Rightarrow \frac{\rho}{\sigma^3} &= \frac{x_u^2+x_l^2}{\sqrt{p_u x_u(x_u-x_l)^3}} \, . \label{eq:tri}
\end{align}
In this expression, $p_u$ remains as a parameter scaling the weight on the edges compared to the weight at $x=0$. Note that for $p_u\to \frac{-x_l}{x_u-x_l}$, we arrive at a binomial distribution while for $p_u\to 0$ we have a delta distribution. From Eq.~\eqref{eq:tri}, we see that in the limit $p_u\to 0$, $\frac{\rho}{\sigma^3} \to \infty$. Thus, we can write for the Berry-Esseen bound
\begin{align}
P[X\leq x_0] &\leq \Phi\left(\sqrt{\frac{M}{\sigma^2}}(x_0-\mu)\right) + \frac{C\rho}{\sigma^3 \sqrt{M}} \leq \frac{C\rho}{\sigma^3 \sqrt{M}} = \frac{C}{\sqrt{M}}\frac{x_u^2+x_l^2}{\sqrt{p_u x_u(x_u-x_l)^3}} \xrightarrow[p_u \to 0]{} \infty \, .
\end{align}
Since we have no information on the actual probability distribution, the largest p-value of the Berry-Esseen bound is $1$. We thus restrict our interests to the Chernoff, Bernstein and Uspensky bounds.

The inequalities of Corollaries \ref{cor:Chern}, \ref{cor:Bern} and \ref{cor:Usp} have certain properties in common: They all depend on the variance $\sigma_0^2$ and the mean value $\mu$. More explicitly, the dependence on $\sigma_0^2$ in all three cases is such that if one increases the variance, the bounds are also increased. We therefore use the following strategy to determine the largest p-values: We increase the variance of an arbitrary random variable in a way leaving the mean value unaffected. We eventually arrive at an easy-to-handle probability distribution with a maximal variance. The bounds resulting from this distribution then serve as upper bounds for all distributions of the same mean value. The bounds given by the concentration inequalities will then only depend on the mean value, so that we can optimize over $\mu$.

Eventually, we show that the largest p-value results from the binomial distribution centered around $\mu=0$. 
\begin{TH1}
Let $X$ be a random variable in the interval $[a,b]$ with an arbitrary probability distribution and with $E[X]=\mu$. Let $X_{bi}$ be a binomially distributed random variable with peaks at the edges $a$ and $b$. Furthermore, $X_{bi}$ has the same mean value as $X$, i.e. $P[X_{bi}=a]=\frac{\mu-b}{a-b}$, $P[X_{bi}=b]=\frac{a-\mu}{a-b}$ and $P[X_{bi}=x]=0$ otherwise. Additionally, let $a\leq 0$. Then
\begin{align*}
\mathrm{Var}[X_{bi}]\geq \mathrm{Var}[X] \, .
\end{align*}
\label{th:Var}
\end{TH1}
\begin{proof}
We define a third random variable $Y$ satisfying 
\begin{align}
P[Y=x]= 
\begin{cases} P[X=x] &\text{if } x \notin dx_0\cup\{a,b\}  \\ 
			  0 & \text{if } x \in dx_0 \\ 
			  P[X=a]+qP[X\in dx_0] &\text{if } x=a \\ 
			  P[X=b]+(1-q)P[X\in dx_0]  &\text{if } x=b 
\end{cases} \, .
\end{align}
Here, $dx_0$ is an infinitesimal set around $x_0\in ]a,b[$. So in other words, the probability distribution function of $Y$ is almost the same as the one of $X$. The only difference is that the set $dx_0$ is "cut out" and the probabilities at $a$ and $b$ are increased (see Fig. \ref{fig:Var}). They are increased in a way which leaves the mean value unaffected, i.e. $q$ is chosen such that $E[Y]=E[X]=\mu$: 
\begin{align}
E[Y] &= E[X]-P[X\in dx_0]x_0+qP[X\in dx_0]a+(1-q)P[X\in dx_0]b \nonumber\\
	 &= \mu + P[X\in dx_0][-x_0+q(a-b)+b] = \mu \nonumber\\
	 &\Leftrightarrow \quad q=\frac{x_0-b}{a-b} \, .
\end{align}
With this and $a\leq 0$, we show that $\mathrm{Var}[Y]\geq \mathrm{Var}[X]$: 
\begin{align}
\mathrm{Var}[Y] &= \mathrm{Var}[X]+P[X\in dx_0]\left[-x_0^2+q(a^2-b^2)+b^2\right] = \mathrm{Var}[X]+P[X\in dx_0]\left[-x_0^2+x_0(a+b)-ab\right]\nonumber\\
	   &\geq \mathrm{Var}[X]+P[X\in dx_0]\left[-x_0^2+x_0(a+x_0)-a x_0\right] \geq \mathrm{Var}[X]+P[X\in dx_0]\left[-x_0^2+x_0^2\right] = \mathrm{Var}[X]  \, .
\end{align}
So we find that $\mathrm{Var}[Y]\geq \mathrm{Var}[X]$, with the equal sign only if $P[X\in dx_0]=0$. By induction, one can gradually "cut out" all the other points in $]a,b[$. During this process, the variance is constantly increased while the mean value remains unchanged. Eventually, one arrives at the random variable $X_{bi}$ with its binomial distribution.  
\begin{figure}
\centering
\includegraphics[width=0.8\textwidth]{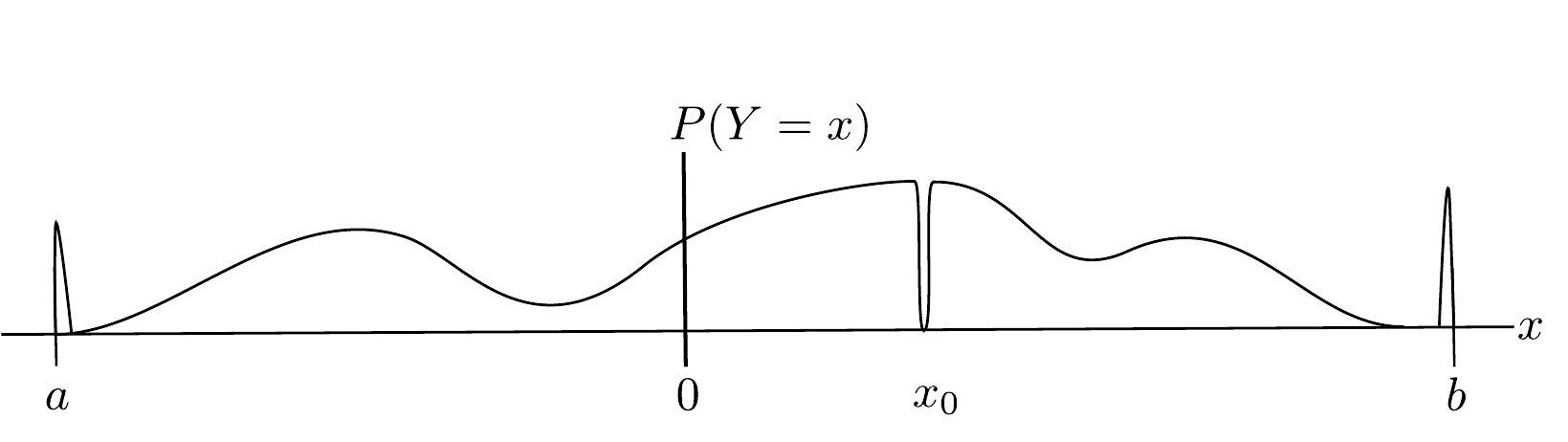}
\caption{Sketch of the probability distribution function of $Y$ in Theorem~\ref{th:Var}.}
\label{fig:Var}
\end{figure}
\end{proof}

By applying the different bounds and using Theorem \ref{th:Var}, we find 
\begin{align}
P[X\leq x_0]\leq 
\begin{cases}
\exp\left(-\frac{(\mu-x_0)^2 M}{4\sigma_{bi}^2(\mu)}\right) =:p_C(\mu,M) &\text{ Chernoff} \\
\\
\exp\left(-\frac{(\mu-x_0)^2 M}{2\sigma_{bi}^2(\mu)+\frac23(x_u-x_l)(\mu-x_0)}\right) =:p_B(\mu,M) &\text{ Bernstein}\\
\\
\frac{\sigma_{bi}^2(\mu)}{\sigma_{bi}^2(\mu)+(x_0-\mu)^2 M} =:p_U(\mu,M) &\text{ Uspensky}
\end{cases} \, ,
\end{align}
where $\sigma_{bi}^2(\mu)=(x_u-\mu)(\mu-x_l)$. As stated above, $\mu\in[0,x_u]$. Restricted to this interval, the three functions $p_C$, $p_B$ and $p_U$ are strictly monotonous decreasing with $\mu$ and therefore take their maximal values at $\mu=0$. This allows us to write
\begin{align}
P[X\leq x_0]\leq 
\begin{cases}
\exp\left(-\frac{x_0^2 M}{4\sigma_{bi}^2}\right) =:p_C(M) &\text{ Chernoff} \\
\\
\exp\left(-\frac{x_0^2 M}{2\sigma_{bi}^2-\frac23(x_u-x_l)x_0}\right) =:p_B(M) &\text{ Bernstein}\\
\\
\frac{\sigma_{bi}^2}{\sigma_{bi}^2+x_0^2 M} =:p_U(M) &\text{ Uspensky}
\end{cases} \, ,
\end{align}
with $\sigma_{bi}^2:=\sigma_{bi}^2(0)=-x_l x_u$.

We identify $M$ with the number of experimental runs, and determine the minimum number of runs required to have $P[X\leq x_0]\leq \varepsilon$. This corresponds to solving $p_i(M)\leq \varepsilon $ to $M$, where $i=C,B,U$.
We find
\begin{align}
M\geq 
\begin{cases}
M_C := \frac{4\sigma_{bi}^2}{x_0^2}\ln\left(\frac{1}{\varepsilon}\right) &\text{ Chernoff} \\
\\
M_B := \frac{2\sigma_{bi}^2-\frac23(x_u-x_l)x_0}{x_0^2}\ln\left(\frac{1}{\varepsilon}\right) &\text{ Bernstein}\\
\\
M_C := \frac{\sigma_{bi}^2(1-\varepsilon)}{\varepsilon x_0^2} &\text{ Uspensky}
\end{cases} \, .
\label{eq:bounds}
\end{align}

\subsection{Comparison of the bounds}
We now compare the three bounds stated in Ineq.~\eqref{eq:bounds} and show that in the regime of interest, the Bernstein bound is the tightest bound.
By comparing $M_B$ to $M_C$, one finds
\begin{align}
M_B\cdot\frac{x_0^2}{\ln(1/\varepsilon)}&=2\sigma_{bi}^2-\frac23(x_u-x_l)x_0<2\sigma_{bi}^2-(x_u-x_l)x_0 \nonumber\\
&<-2x_l x_u-x_u x_l-x_l x_u =-4x_l x_u = 4\sigma_{bi}^2= M_C\cdot\frac{x_0^2}{\ln(1/\varepsilon)} \nonumber\\
\Rightarrow \, M_B&<M_C \, .
\end{align}
This means the bound resulting from Bernstein's inequality is tighter than the Chernoff bound for all values of $x_l$, $x_u$ and $x_0$.

Now we compare $p_B$ to $p_U$. Plotting both functions reveals that Uspensky's inequality is better for small numbers of experimental runs, whereas Bernstein's is better for larger numbers of runs. We now want to estimate for which $\varepsilon$ both approximations require the same number of runs $M$. We therefore set $M_B=M_U$:
\begin{align}
\frac{2\sigma_{bi}^2-\frac23(x_u-x_l)x_0}{x_0^2}\ln\left(\frac{1}{\varepsilon}\right) =\frac{\sigma_{bi}^2(1-\varepsilon)}{\varepsilon x_0^2} \,\Leftrightarrow\, 2+\frac23 \frac{(x_u-x_l)(-x_0)}{\sigma_{bi}^2} = \frac{1-\varepsilon}{\varepsilon\ln\left(\frac{1}{\varepsilon}\right)}  \, . \label{eq:intersect}
\end{align}
The function on the right side of Eq. \eqref{eq:intersect} is strictly monotonous decreasing with $\varepsilon$. The left side is just a number. So the minimal $\varepsilon$ possible is achieved if the left side is maximized. We thus make the following estimation:
$$\frac{(x_u-x_l)\vert x_0\vert}{\sigma_{bi}^2}=\frac{(x_u+\vert x_l\vert)\vert x_0\vert}{\vert x_l\vert x_u} \leq \frac{x_u\vert x_l\vert + \vert x_l\vert x_u}{\vert x_l\vert x_u} = 2 \, .$$
For this value, Eq.~\eqref{eq:intersect} yields $\varepsilon\approx 0.127$. 
This means that the Uspensky bound can only be better than the Bernstein bound for $\varepsilon\geq 0.127$. Since we are interested in probabilities $\varepsilon\leq 0.1$, the Bernstein bound remains the preferred one.

\subsection{Statistical optimization}
\label{sec:exp}
Let us now present the result of numerical studies on the number of measurement runs allowing one to rule out non-Bell-correlated states. Here we rely on the Bernstein bound presented in Ineq.~\eqref{eq:bounds}.

\subsubsection{Choice of settings}
First, we study the choice of measurement settings, i.e. the choice of $\nu$ in Eq.~\eqref{eq:WInf}, which allow for the strongest statistical claim. For this, we minimize the number of experimental runs sufficient to conclude about the presence of a Bell-correlated state with given confidence $1-\varepsilon$ over the choices of $\nu$. We then compare this number of measurements $M$ to the one obtained when choosing the $\nu=\text{sinh}\left(\frac{\C_b}{1-\zeta_a^2}\right)$, i.e. for the settings which maximize the witness value (see Eq.~\eqref{eq:bestNu}). 

The result of this comparison is presented in Figure~\ref{fig:VarNu} for a particular choice of $\C_b$ and $\zeta_a^2$. Clearly, it is advantageous to reoptimize the measurement setting in order to maximize the statistical evidence, and thus a different witness should be considerd in this case. In the rest of this appendix, measurement settings are optimized in order to maximize statistical evidence.

\begin{figure}
\begin{minipage}{0.49\textwidth}
\centering
\includegraphics[width=0.95\textwidth]{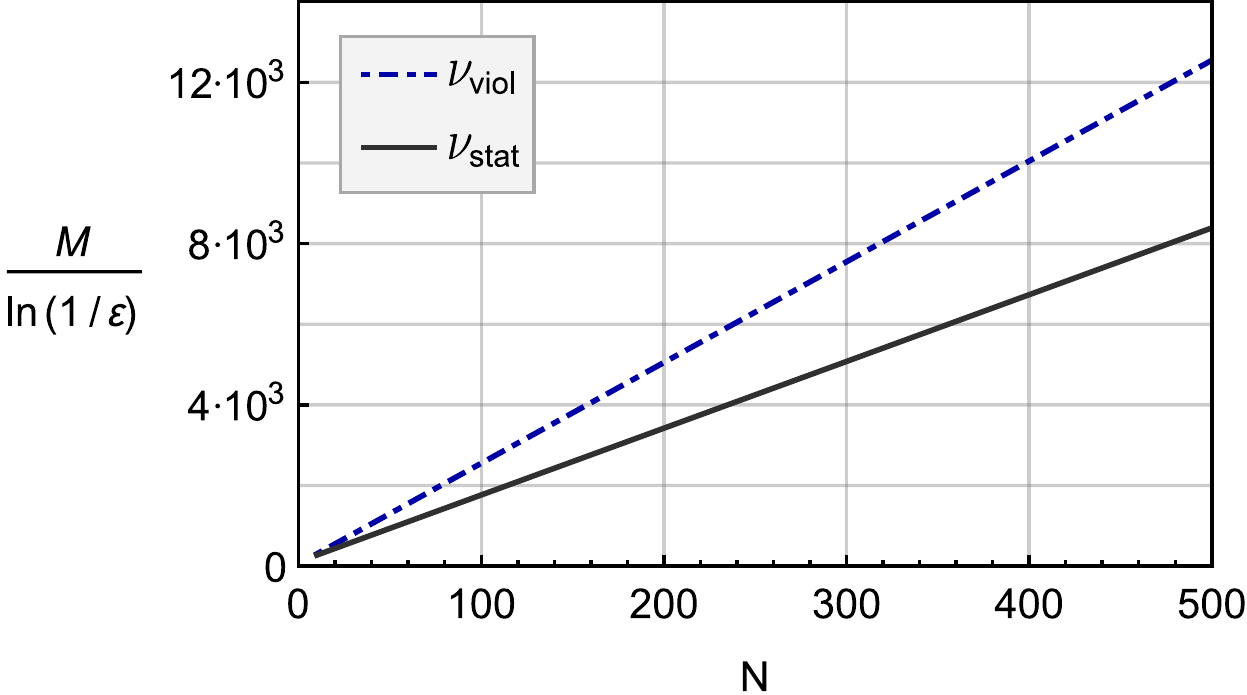}
\caption{Plot of the required number of runs as a function of the number of atoms, with $\C_b=0.98$ and $\zeta_a^2$. The blue dashdotted line results from the $\nu$ which maximizes the violation of witness \eqref{eq:finalW} whereas the solid black line corresponds to the $\nu$ minimizing the number of runs. As one can see, the optimization of $\nu$ in terms of statistics reduces the number of runs by a factor of approximately~$\frac23$.}
\label{fig:VarNu}
\end{minipage}
\hfill
\begin{minipage}{0.49\textwidth}
\centering
\includegraphics[width=0.95\linewidth]{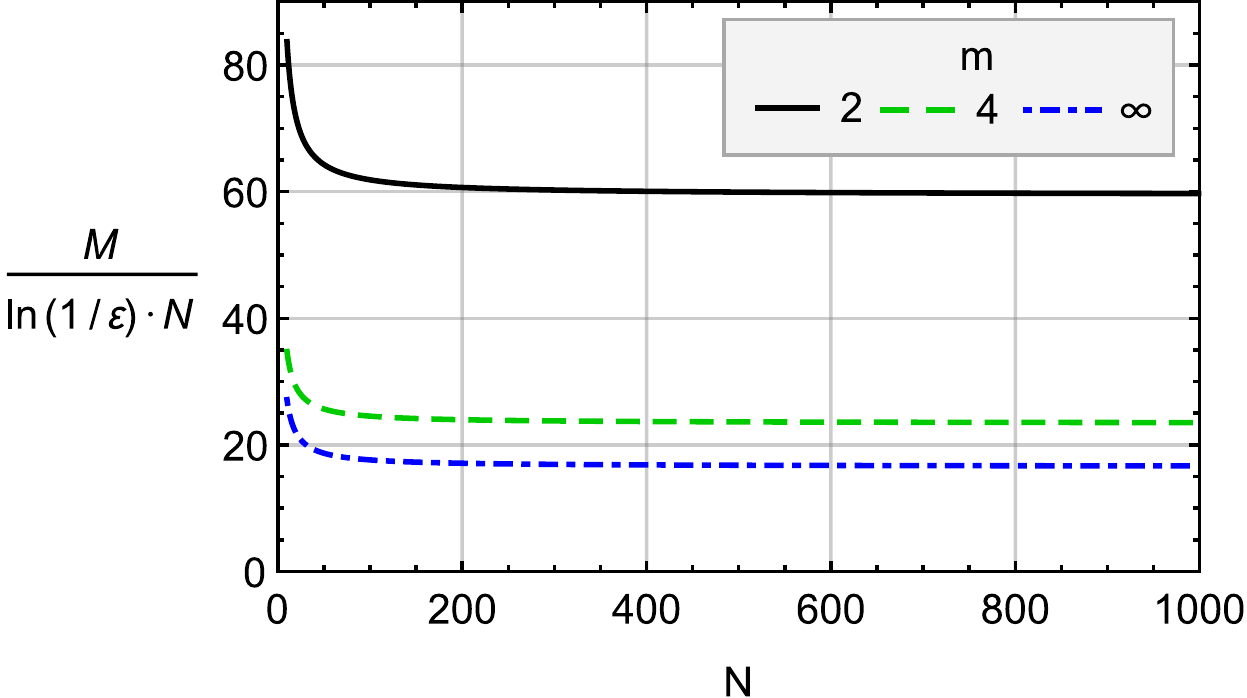}
\caption{Plots of the required number of experimental runs $M$ per spin for a confidence of $\varepsilon$ in units of $\ln(\frac1\varepsilon)$, as a function of the number of spins $N$. The plots are for $m=2,4,\infty$, with $\zeta_a^2=0.272$ and $\C_b=0.98$. The ratios tend to constants for larger systems.}
\label{fig:RatioMN}
\end{minipage}
\end{figure}

\subsubsection{Linear relation between the number of measurement runs and the number of spins}
Figure~\ref{fig:RatioMN} depicts the number of measurement runs per spin needed in order to reach a confidence of $1-\varepsilon$. The plots are for the values of $\C_b$ and $\zeta_a^2$ stated in~\cite{Schmied16}, and for $m=2,4,\infty$ settings. As explained in the main text, we observe that the ratio $\frac{M}{N}$ tends to a constant as $N$ increases. 
This plot also illustrates the gain obtained in using the newly derived Bell inequality with $m$ settings.

\subsubsection{Minimum squeezing requirement for finite $M$}
As discussed in the main text, violation of our witness requires a finite amount of squeezing for systems of finite size $N$. Here we study this requirement when finite number of measurement runs are taken into account. In Fig. \ref{fig:Ca2vsN} we show the upper bound on $\zeta_a^2$ as a function of the number of spins $N$ and the number of measurement runs $M$ for the case $\C_b=0.98$, $\varepsilon=0.1$. Although Fig. 2 in the main text shows that a violation with $N=1000$ spins is possible when $\zeta_a^2 \lesssim 0.8$, we see here that a squeezing of $\zeta_a^2 \lesssim 0.5$ is needed if the number of measurement is limited to $10^6$.

\begin{figure}
\centering
\includegraphics[width=0.43\textwidth]{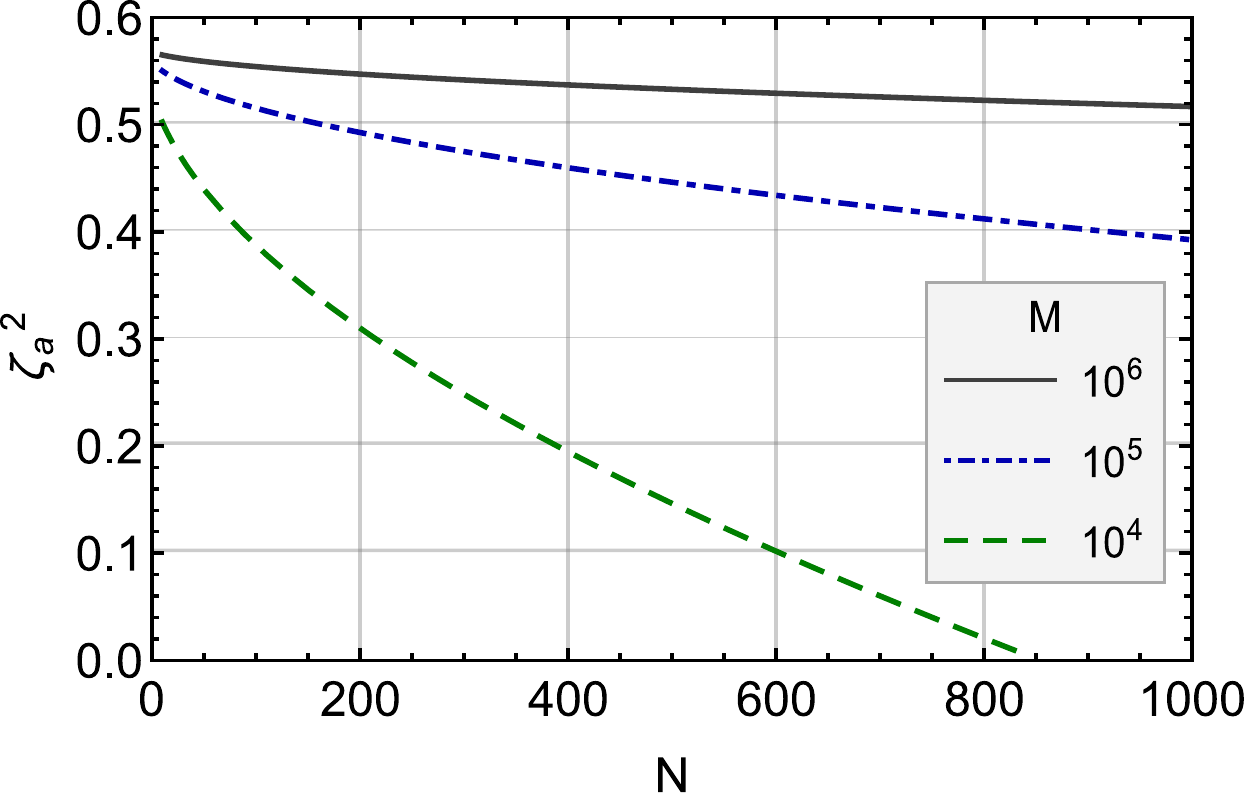}
\caption{Plots of the required scaled second moment $\zeta_a^2$ as a function of the number of atoms $N$, with $\C_b=0.98$ and $\varepsilon=0.1$. The different curves represent different numbers of experimental runs ($M=10^4,10^5,10^6$).}
\label{fig:Ca2vsN}
\end{figure}

\end{document}